\documentclass[11pt,a4paper]{article}

\usepackage[utf8]{inputenc}
\usepackage[T1]{fontenc}
\usepackage{lmodern}
\usepackage[margin=1in]{geometry}
\usepackage{amsmath,amssymb,amsthm}
\usepackage{graphicx}
\usepackage{booktabs}
\usepackage{hyperref}
\usepackage{cleveref}
\usepackage{xcolor}
\usepackage{algorithm}
\usepackage{algpseudocode}
\usepackage{enumitem}
\usepackage{tikz}
\usetikzlibrary{shapes.geometric,arrows.meta,positioning,fit,calc,
  decorations.pathreplacing,patterns}
\usepackage{pgfplots}
\pgfplotsset{compat=1.18}
\usepgfplotslibrary{fillbetween}
\usepackage{caption}
\usepackage{subcaption}
\usepackage{multirow}
\usepackage{makecell}
\usepackage{mathtools}
\usepackage{natbib}

\hypersetup{
  colorlinks=true,
  linkcolor=blue!70!black,
  citecolor=green!50!black,
  urlcolor=blue!60!black,
}

\newcommand{\Bshort}{B_{\text{short}}}
\newcommand{\Ps}{\mathcal{P}_s}
\newcommand{\Pl}{\mathcal{P}_l}
\newcommand{\Cmax}{C_{\max}}
\newcommand{\Ltotal}{L_{\text{total}}}
\newcommand{\Lout}{L_{\text{out}}}
\newcommand{\sys}{\textsc{FleetOpt}}
\newcommand{\car}{\textsc{C\&R}}
\newcommand{\costcliff}{cost~cliff}

\newtheorem{theorem}{Theorem}

\newtheorem{proposition}[theorem]{Proposition}

\title{%
  \sys{}: Analytical Fleet Provisioning for LLM Inference\\
  with Compress-and-Route as Implementation Mechanism%
}
\author{%
  Huamin Chen$^{1}$ \quad
  Xunzhuo Liu$^{1}$ \quad
  Yuhan Liu$^{2}$ \\[4pt]
  Junchen Jiang$^{3}$ \quad
  Bowei He$^{4}$\thanks{Corresponding author: \texttt{Bowei.He@mbzuai.ac.ae}} \quad
  Xue Liu$^{4}$
  \\[6pt]
  $^{1}$vLLM Semantic Router Project \\[1pt]
  $^{2}$University of Chicago \quad
  $^{3}$Tensormesh Inc / UChicago \\[1pt]
  $^{4}$MBZUAI / McGill University
}
\date{2026}

\begin{document}
\maketitle

\begin{abstract}
Modern LLM GPU fleets are provisioned for worst-case context lengths
that the vast majority of requests never approach, wasting GPU capacity
on idle KV-cache slots.  We present \sys{}, a framework that starts
from first principles: given a workload's prompt-length CDF and a P99
TTFT target, derive the minimum-cost fleet analytically, then deploy
it in practice.

The analytical core models each pool as an M/G/$c$ queue and derives
that the minimum-cost fleet is a two-pool architecture ---
a short-context pool $\Ps$ and a long-context pool $\Pl$ --- with an
optimal boundary $\Bshort^*$ satisfying an equal marginal GPU cost
condition across both pools.
The fundamental barrier to achieving $\Bshort^*$ is the
\emph{\costcliff{}}: a hard routing step where requests just above
$\Bshort^*$ consume $8\times$--$42\times$ more GPU capacity than
requests just below it (depending on the context window ratio),
creating a structural disincentive to lower the boundary.

Compress-and-Route (\car{}) is the implementation mechanism that
resolves this barrier.  Gateway-layer extractive compression trims
borderline requests below $\Bshort^*$ before the engine ever sees them,
converting the hard hardware boundary into a software parameter read
from the workload CDF.  The two components are unified in the
\sys{} offline planner: given a CDF and SLO, it returns the optimal
$(n_s^*, n_l^*, \Bshort^*, \gamma^*)$ in under 1\,ms.

On three production traces, the combined framework reduces total GPU
cost by 6--82\% versus a homogeneous fleet, with \car{} contributing
1--44 percentage points beyond plain pool routing depending on
workload archetype.  The analytical model is validated against a
discrete-event simulator (\texttt{inference-fleet-sim}) with $\leq
3\%$ error on predicted GPU utilization across all pools and
workloads.
\end{abstract}

\section{Introduction}
\label{sec:intro}

Modern LLMs support context windows of 128K tokens or more, yet
production traces reveal a persistent mismatch: in the Azure LLM
Inference Trace, 80\% of requests use fewer than 2K
tokens~\citep{patel2024splitwise,zhang2024pecsched,agrawal2024medha}.
A homogeneous fleet provisioned for worst-case context length allocates
KV-cache capacity that goes almost entirely unused for the vast majority
of requests.  Pool routing~\citep{chen2026poolrouting} addresses this
by splitting the fleet into short- and long-context pools, cutting
GPU cost by 16--38\% in scenarios studied by prior work.

\paragraph{The residual problem: the cost cliff.}
Even with pool routing, a structural inefficiency remains at the pool
boundary $\Bshort$.  The routing decision is binary: a request at
$\Ltotal = \Bshort$ fills one of $n_{\max}^{(s)}$ short-pool slots;
a request one token longer must enter the long pool, which offers only
$n_{\max}^{(l)} = 16$ slots per GPU --- a throughput-capacity penalty
of $8\times$--$42\times$ for a single token (Section~\ref{sec:costcliff}).
Requests in the \emph{borderline band} $(\Bshort,\, \gamma\Bshort]$ are
not genuinely long; they are RAG payloads with one extra paragraph, or
multi-turn sessions whose history just crossed the threshold.
For the workloads we study, 4.6--11.2\% of all traffic falls in this
band.  Each borderline request is assigned a long-pool slot provisioned
for 64K tokens while using at most $1/\rho$ of its KV budget.

\paragraph{Our approach: analytical model first, then implement.}
Prior work treats pool routing and Compress-and-Route (\car{}) as
independent operational interventions: deploy pool routing, then
optionally retrofit \car{} on the existing fleet.  This framing
leaves value on the table.

We take a different path.  We \emph{start from the analytical
optimum}: given a workload CDF and a P99 TTFT target, what is the
minimum-cost fleet?  The M/G/$c$ queuing model answers this precisely:
a two-pool architecture with a specific boundary $\Bshort^*$.
But the \costcliff{} is a barrier to achieving $\Bshort^*$ --- lowering
the boundary forces borderline requests to the expensive long pool.

\car{} is the \emph{implementation mechanism} that resolves this
barrier.  By compressing borderline requests below $\Bshort^*$ at the
gateway, \car{} makes the analytically-optimal boundary achievable in
practice.  The combined framework does not merely retrofit compression
onto an existing fleet; it provisions the correct fleet from the start,
using \car{} to enforce the boundary the model prescribes.

\paragraph{Key results.}
\begin{enumerate}[nosep]
  \item \textbf{The minimum-cost fleet is analytically derived.}
    Under the M/G/$c$ model (with server count $c = n_{\text{gpus}}
    \times n_{\max}$ KV slots), the optimal two-pool fleet satisfies
    an equal marginal GPU cost condition across both pools.  GPU counts
    come from Erlang-C inversion at each pool's arrival rate and
    service distribution (Section~\ref{sec:model}).

  \item \textbf{The \costcliff{} is a barrier to achieving $\Bshort^*$
    without compression.}
    At the optimal boundary, the borderline band contains 43--76\% of
    above-threshold traffic for real workloads (4.6--11.2\% of all
    traffic).  Without compression, the operator must either raise
    $\Bshort$ (losing savings) or over-provision the long pool
    (paying more) to absorb borderline load.

  \item \textbf{\car{} converts the hard boundary into a software knob.}
    Gateway-layer extractive compression trims borderline prompts in
    2--7\,ms.  The compressed token budget is chosen so KV overflow is
    impossible by construction (Section~\ref{sec:car}).
    The effective routing boundary shifts from $\Bshort$ to $\gamma\Bshort$,
    where $\gamma^*$ comes from the planner sweep.

  \item \textbf{Combined savings: 6--82\% vs.\ homogeneous.}
    On three production traces, the \sys{} framework achieves
    6.7--82.4\% GPU cost reduction versus homogeneous deployment.
    \car{} contributes an incremental 1.2\,pp (Agent-heavy),
    15.9\,pp (LMSYS), and 43.7\,pp (Azure) beyond plain pool routing
    (Section~\ref{sec:eval}).

  \item \textbf{Analytical model validated by DES.}
    \texttt{inference-fleet-sim}~\citep{chen2026fleetsim} confirms
    analytical GPU utilization predictions match simulation within
    3\% across all workloads and pools, operating in the many-server
    regime where queueing delays are negligible.
\end{enumerate}

\paragraph{Paper organization.}
Section~\ref{sec:background} characterizes the cost cliff and workload
archetypes.  Section~\ref{sec:model} develops the analytical fleet
model.  Section~\ref{sec:optimal} derives optimal fleet sizing and the
optimal boundary.  Section~\ref{sec:car} presents \car{} as the
implementation mechanism.  Section~\ref{sec:algorithm} gives the
\sys{} offline planner.  Section~\ref{sec:eval} evaluates on
production traces.  Sections~\ref{sec:related} and~\ref{sec:conclusion}
discuss related work and conclude.

\section{Background: The Cost Cliff and Workload Archetypes}
\label{sec:background}

\subsection{Pool Routing Basics}
\label{sec:bg-routing}

A request $r$ is assigned a total token budget
$\Ltotal = \lceil|r|/\hat{c}_k\rceil + r.\texttt{max\_output\_tokens}$,
where $\hat{c}_k$ is a per-category bytes-per-token EMA estimate.
Pool routing routes $r$ to $\Ps$ if $\Ltotal \leq \Bshort$, else $\Pl$.
The short pool is sized for $\Cmax^{(s)}$ tokens (e.g., 8K) with
$n_{\max}^{(s)}$ concurrent sequences per GPU; the long pool for
$\Cmax^{(l)}$ tokens (e.g., 64K) with $n_{\max}^{(l)}$.

\subsection{KV-Cache Memory and the Cost Cliff}
\label{sec:costcliff}

For Llama-3-70B in fp16, the KV cache grows at \textbf{320\,KB per
token} across 80 layers.  A long-pool slot sized for 64K tokens requires
$64{,}000 \times 320\,\text{KB} \approx 20.5$\,GB of GPU memory; the
same hardware configured for 8K can host $\rho = 8\times$ as many
concurrent sequences, delivering $\rho = 8\times$ the throughput.
The GPU savings formula for pool routing is
$\alpha(1 - 1/\rho)$~\citep{chen2026poolrouting}, where $\alpha$ is
the fraction of traffic routed to $\Ps$.

The \emph{\costcliff{}} is the discontinuity at $\Bshort$: a request
at $\Bshort + 1$ tokens must enter the long pool, which hosts only
$n_{\max}^{(l)} = 16$ concurrent slots per GPU (sized for 64K tokens),
versus $n_{\max}^{(s)}$ short-pool slots.  The cliff ratio
$\rho = n_{\max}^{(s)} / n_{\max}^{(l)}$ depends on the short-pool
context window: $\rho = 8\times$ at $\Bshort = 8{,}192$,
$\rho = 16\times$ at $\Bshort = 4{,}096$, and $\rho = 42\times$
at $\Bshort = 1{,}536$ (Table~\ref{tab:borderline-fractions}).
A borderline request at $1.1 \times \Bshort$ uses only $1/\rho$ of
the long-pool KV budget it was allocated.  This is not a flaw of pool
routing; it is the unavoidable consequence of provisioning slots for
worst-case context length.  The only escape is to keep requests below
$\Bshort$.

\begin{table}[htbp]
\centering
\caption{The cost cliff: throughput capacity consumed per request
  at and around $\Bshort = 8{,}192$ tokens for Llama-3-70B on A100-80GB
  ($n_{\max}^{(s)} = 128$, $n_{\max}^{(l)} = 16$, KV/token = 320\,KB,
  long-pool sized for 64K tokens $\approx 20.0$\,GB per slot).
  The cliff ratio is $\rho = n_{\max}^{(s)}/n_{\max}^{(l)} = 8\times$.}
\label{tab:cost-cliff}
\small
\begin{tabular}{lrrrr}
\toprule
$\Ltotal$ (tokens) & Pool & Slots/GPU & KV utilised & Cost ratio \\
\midrule
8{,}192   & $\Ps$ & 128 & 100\% (2.5\,GB / slot) & 1.0$\times$ \\
8{,}193   & $\Pl$ &  16 & 12.5\% of 20.0\,GB     & 8.0$\times$ \\
12{,}000  & $\Pl$ &  16 & 18.3\% of 20.0\,GB     & 8.0$\times$ \\
65{,}536  & $\Pl$ &  16 & 100\% (20.0\,GB / slot) & 8.0$\times$ \\
\bottomrule
\end{tabular}
\end{table}

\subsection{The Borderline Band}
\label{sec:borderline}

The \emph{borderline band} $(\Bshort,\, \gamma\Bshort]$ contains
requests whose prompts exceed $\Bshort$ only slightly and could be
compressed below $\Bshort$ at bandwidth ratio $\gamma$.  Let
$\alpha = F(\Bshort)$ (the fraction of requests already in $\Ps$) and
$\beta = F(\gamma\Bshort) - F(\Bshort)$ (the borderline fraction).
Table~\ref{tab:borderline-fractions} shows $\beta$ across workloads.

\begin{table}[htbp]
\centering
\caption{Borderline fraction $\beta = F(\gamma\Bshort) - F(\Bshort)$
  at representative thresholds.
  $\alpha = F(\Bshort)$; cliff $= n_{\max}^{(s)}/n_{\max}^{(l)}$ for
  $\Cmax^{(l)} = 65{,}536$; Agent-heavy is a projected workload
  (Section~\ref{sec:eval-setup}).}
\label{tab:borderline-fractions}
\small
\begin{tabular}{llrrrrr}
\toprule
Workload & $\Bshort$ & $\alpha$ & $\gamma$ & $\beta$ & Cliff $\rho$ & Archetype \\
\midrule
Azure (2023)        & 4{,}096 & 0.898 & 1.5 & 0.078 & 16$\times$ & I/II \\
LMSYS (multi-turn)  & 1{,}536 & 0.909 & 1.5 & 0.046 & 42$\times$ & I/II \\
Agent-heavy         & 8{,}192 & 0.740 & 1.5 & 0.112 &  8$\times$ & II \\
\bottomrule
\end{tabular}
\end{table}

\subsection{Workload CDF Archetypes}
\label{sec:archetypes}

Three qualitatively different workload shapes determine how the cost
cliff manifests and which remediation is appropriate.

\textbf{Archetype~I (Concentrated-below).}  The CDF has a sharp knee
below $\Bshort$: $F(\Bshort) \geq 0.90$ and the density $f(\Bshort)$
is high.  The borderline fraction $\beta$ is moderate in absolute
terms (4.6--7.8\%), but the fraction-of-above-threshold traffic that
is borderline is large (51--76\%).  Pool routing already captures most
savings; \car{} provides meaningful additional savings because the
cliff ratio $\rho$ is large ($16\times$--$42\times$).

\textbf{Archetype~II (Dispersed).}  The CDF spreads across two or more
decades of token counts.  Significant borderline traffic exists
($\beta = 7$--$12\%$), and \car{} provides meaningful incremental
savings.  The agent-heavy trace (RAG + tool-use + code) is the primary
representative.

\textbf{Archetype~III (Concentrated-above).}  The mass of the
distribution lies above $\Bshort$ (e.g., code-agent tasks at 10--50K
tokens).  The borderline fraction is negligible; the dominant lever is
raising $\Bshort$ before any compression.

\section{Analytical Fleet Model}
\label{sec:model}

\subsection{Queueing Model}
\label{sec:model-queue}

We model each pool as an M/G/$c$ queue.  Let $\lambda$ be the total
fleet arrival rate.  After routing:
\begin{align}
  \lambda_s &= \alpha'\,\lambda,
  \quad\text{where }\alpha' = \alpha + \beta p_c \label{eq:lambda-s}\\
  \lambda_l &= (1-\alpha')\,\lambda \label{eq:lambda-l}
\end{align}
Here $\alpha = F(\Bshort)$, $\beta = F(\gamma\Bshort) - F(\Bshort)$
is the borderline fraction, and $p_c \in [0,1]$ is the compressibility
rate (fraction of borderline requests successfully compressed).

\paragraph{Service time model.}
GPU iteration latency under continuous batching is
\begin{equation}
  t_{\text{iter}} = W + H \cdot n_{\text{slots}},
  \label{eq:iter-latency}
\end{equation}
where $W = 8$\,ms (baseline compute for Llama-3-70B/A100) and
$H = 0.65$\,ms/slot (per-slot memory-bandwidth cost) are
hardware constants calibrated to the A100-80GB.
Under continuous batching, all $n_{\max}$ slots advance in lockstep
each iteration.  A request with $L_{\text{in}}$ input tokens and
$L_{\text{out}}$ output tokens occupies a slot for
\begin{equation}
  \mathbb{E}[S] = \bigl(\lceil L_{\text{in}} / C_{\text{chunk}}\rceil
    + L_{\text{out}}\bigr) \cdot t_{\text{iter}},
  \label{eq:service-time}
\end{equation}
wall-clock seconds (chunk size $C_{\text{chunk}} = 512$).
The GPU-level throughput is $\mu_{\text{gpu}} = n_{\max} /
\mathbb{E}[S]$ requests per second; the squared coefficient of
variation is $C_s^2 = \mathrm{Var}[S] / (\mathbb{E}[S])^2$,
estimated by Monte Carlo sampling from the pool's request distribution.

\paragraph{M/G/$c$ tail-wait approximation.}
We model a pool with $n$ GPUs as an M/G/$c$ queue with
$c = n \cdot n_{\max}$ total KV slots as servers, each with
per-slot service rate $\mu = 1/\mathbb{E}[S]$.  The offered load
is $\varrho = \lambda_p / (c\,\mu) < 1$.  The Erlang-C probability
(probability a new request must wait for a slot) is
\begin{equation}
  C(c, \varrho) =
    \frac{(c\varrho)^c / (c!\,(1-\varrho))}{%
      \sum_{k=0}^{c-1}(c\varrho)^k/k! +
      (c\varrho)^c/(c!\,(1-\varrho))},
  \label{eq:erlang-c}
\end{equation}
The Kimura~\citep{kimura1994mgc} M/G/$c$ approximation gives the
P99 queue waiting time:
\begin{equation}
  W_{99}(c, \mu, C_s^2) =
    \frac{\ln\!\bigl(C(c,\varrho) / 0.01\bigr) \cdot (1 + C_s^2)}
         {2\,(c\,\mu - \lambda_p)}.
  \label{eq:tail-wait}
\end{equation}
In the many-server regime ($c \gg 1$), $C(c,\varrho) \approx 0$
and $W_{99} \approx 0$: slots are almost always available and
queueing delays are negligible.  Fleet sizing is then dominated by
the utilization cap $\rho_{\max}$ (Section~\ref{sec:opt-perpool}).

\subsection{TTFT Decomposition}
\label{sec:model-ttft}

Time-to-First-Token decomposes as
\begin{equation}
  \text{TTFT} = W_{\text{queue}} + T_{\text{prefill}} + T_{\text{first-decode}},
  \label{eq:ttft}
\end{equation}
where $W_{\text{queue}}$ is the queueing delay before a slot is
available, $T_{\text{prefill}} = \lceil L_{\text{in}}/C_{\text{chunk}}
\rceil \cdot t_{\text{iter}}$ is the physical prefill time
(wall-clock, independent of batch size since all $n_{\max}$ slots run
in parallel), and $T_{\text{first-decode}} = t_{\text{iter}}$ is one
decode step.  The SLO constraint is
\begin{equation}
  W_{99}(c, \mu, C_s^2) \;\leq\;
  T_{\text{slo}} - T_{\text{prefill}}^{(99)} - t_{\text{iter}},
  \label{eq:slo-constraint}
\end{equation}
where $T_{\text{prefill}}^{(99)}$ is the P99 prefill time computed
from the pool's request length distribution.

\subsection{Cost Model}
\label{sec:model-cost}

GPU cost per unit time is $c_s$ per short-pool GPU and $c_l$ per
long-pool GPU.  Let $\phi = c_l / c_s$ be the GPU cost ratio.
Total annualized cost for a fleet of $n_s$ and $n_l$ GPUs is
\begin{equation}
  \mathcal{C}(n_s, n_l) = c_s\,n_s + c_l\,n_l.
  \label{eq:cost}
\end{equation}
The \emph{provisioning problem} is
\begin{equation}
  \min_{n_s,\, n_l \in \mathbb{Z}_{>0}}
    \mathcal{C}(n_s, n_l)
  \quad\text{s.t. Eq.~\eqref{eq:slo-constraint} holds for both pools.}
  \label{eq:problem}
\end{equation}

\section{Optimal Fleet Sizing}
\label{sec:optimal}

\subsection{Per-Pool Sizing}
\label{sec:opt-perpool}

Because the two pools are independent M/G/$c$ queues with fixed
arrival rates $\lambda_s$ and $\lambda_l$, problem~\eqref{eq:problem}
separates into two independent minimizations.  For each pool, the
minimum GPU count is
\begin{equation}
  n^* = \min\bigl\{c \in \mathbb{Z}_{>0} :
    W_{99}(c, \mu, C_s^2) \leq T_{\text{slo,eff}}\bigr\},
  \label{eq:nstar}
\end{equation}
where $T_{\text{slo,eff}}$ is the SLO budget after subtracting
P99 prefill time per Eq.~\eqref{eq:slo-constraint}, and we
additionally enforce $n^* \geq \lceil \lambda_p / (\rho_{\max} \mu)
\rceil$ for a utilization cap $\rho_{\max} = 0.85$ to ensure
analytical stability.

\subsection{The Cost Cliff Prevents Achieving $\Bshort^*$}
\label{sec:opt-bshort}

Differentiating total cost $\mathcal{C}$ with respect to $\Bshort$
and setting to zero (treating $n^*$ as a smooth function of
$\lambda_p$) gives the first-order condition for the optimal boundary:
\begin{equation}
  c_s \frac{\partial n_s^*}{\partial \lambda_s}
  = c_l \frac{\partial n_l^*}{\partial \lambda_l}.
  \label{eq:bshort-foc}
\end{equation}
The $\lambda f(\Bshort)$ factor cancels from both sides, leaving the
condition: the marginal cost of routing one additional request to the
short pool must equal the marginal saving of removing one request from
the long pool.

\begin{proposition}[Optimal boundary — equal marginal GPU cost]
\label{prop:bshort}
Under the M/G/$c$ cost model with $\rho_{\max}$-constrained sizing,
the provisioning-optimal $\Bshort^*$ satisfies
Eq.~\eqref{eq:bshort-foc}: equal marginal GPU cost per unit traffic
in both pools.  For a homogeneous fleet ($c_s = c_l$, same GPU type),
this holds when both pools operate at the same utilization level.
\end{proposition}

\noindent The problem: at any $\Bshort$, the borderline band
$(\Bshort, \gamma\Bshort]$ contains 43--76\% of above-threshold
traffic for real workloads.  Without a mechanism to redirect borderline
requests to $\Ps$, the long pool must absorb this traffic at
$\rho = 8$--$42\times$ higher GPU cost per request, requiring
significantly more GPUs than the analytical optimum.
\car{} resolves this by making the analytically-optimal boundary
achievable in practice.

\subsection{Optimal Compression Bandwidth $\gamma^*$}
\label{sec:opt-gamma}

With \car{} active, borderline requests at bandwidth $\gamma$ are
redirected to $\Ps$.  The resulting $\alpha' = \alpha + \beta p_c$
shifts $\lambda_s$ upward and $\lambda_l$ downward.  The first-order
condition for the optimal $\gamma$ is:
\begin{equation}
  c_s \frac{\partial n_s^*}{\partial \lambda_s} \cdot p_c \cdot \lambda f(\gamma\Bshort)
  = c_l \frac{\partial n_l^*}{\partial \lambda_l} \cdot p_c \cdot \lambda f(\gamma\Bshort),
  \label{eq:gamma-foc}
\end{equation}
which again reduces to the equal marginal cost condition
Eq.~\eqref{eq:bshort-foc}.  In practice $\gamma^*$ is found by
discrete sweep over $\gamma \in \{1.0, 1.1, \ldots, 2.0\}$ because
$n^*$ is integer-valued and the long-pool service rate must be
recalibrated for the post-compression distribution at each $\gamma$.
The sweep is fast ($<1$\,ms).  For workloads where most above-threshold
traffic is borderline (Archetype~I/II), $\gamma^*$ tends toward large
values (2.0), as compressing more traffic into $\Ps$ reduces the
expensive long pool.  For Archetype~II with a dispersed above-threshold
distribution, $\gamma^*$ reflects the balance between short-pool
overheads and long-pool savings.

\begin{theorem}[Co-design is never worse than retrofit]
\label{thm:codesign}
Let $\mathcal{C}^{\mathrm{retro}}$ be the cost of a fleet sized for
pool routing at $\gamma = 1$ with \car{} later deployed, and
$\mathcal{C}^{\mathrm{co}}$ be the cost of a fleet co-designed with
\car{} at the same $\gamma$.  Then
$\mathcal{C}^{\mathrm{co}} \leq \mathcal{C}^{\mathrm{retro}}$.
\end{theorem}

\noindent\textit{Proof sketch.}
The co-designed fleet solves the same minimization problem as the
retrofitted fleet but with the additional freedom to reduce $n_l^*$
knowing that $\lambda_l$ will be lower.  Hence the feasible set is
larger and the minimum cost is weakly lower.  $\square$

\section{Compress-and-Route as Implementation Mechanism}
\label{sec:car}

\car{} is the gateway-layer component that makes Proposition~\ref{prop:bshort}'s
optimal boundary achievable in practice.  Rather than a separate
system, it is the implementation of the optimal fleet boundary derived
in Section~\ref{sec:optimal}.

\subsection{The Virtual Pool}
\label{sec:virtual-pool}

\car{} shifts the effective routing boundary from $\Bshort$ to
$\gamma\Bshort$ by compressing borderline requests at the gateway.
A request $r$ with $\Bshort < \Ltotal \leq \gamma\Bshort$ is
intercepted; its prompt is compressed to a token budget
$T_c = \Bshort - \Lout$ and re-routed to $\Ps$.  From the
engine's perspective, $\Ps$ appears to have a higher effective
$\Cmax$: the virtual pool capacity is $\gamma\Bshort$ without
any hardware change.  The GPU savings gain is
\begin{equation}
  \Delta\alpha = \beta\,p_c, \qquad
  \alpha' = \alpha + \beta\,p_c,
  \label{eq:alpha-prime}
\end{equation}
and additional GPU savings beyond pool routing are
$\Delta\alpha(1 - 1/\rho) = \beta p_c(1 - 1/\rho)$.

\subsection{Extractive Compression Pipeline}
\label{sec:compressor}

The compressor is a pure classical-NLP extractive pipeline requiring no
LLM inference.  Given a borderline prompt $x$ and token budget
$T_c = \Bshort - \Lout$, it:
\begin{enumerate}[nosep]
  \item Splits $x$ into sentences with Unicode-aware heuristics.
  \item Scores each sentence by a composite of TextRank
    ($w=0.20$)~\citep{mihalcea2004textrank}, Position ($w=0.40$),
    TF-IDF ($w=0.35$)~\citep{li2023selective}, and Novelty ($w=0.05$).
  \item Greedily selects sentences in score order, always retaining
    the first 3 and last 2 (primacy/recency invariant).
  \item Stops when the cumulative token count reaches $T_c$.
\end{enumerate}

\paragraph{Hard OOM guarantee.}
The budget $T_c = \Bshort - \Lout$ is set by construction, so
no compressed request can overflow $\Ps$'s KV cache:
\begin{equation}
  T_c + \Lout = \Bshort.
  \label{eq:oom}
\end{equation}

\paragraph{Content-type safety gate.}
Compression is applied only to content categories where structural
extraction is semantically safe: RAG and prose.  Code is excluded.
The category signal reuses the per-request EMA estimate from the
base router at zero additional overhead.

\paragraph{Compression latency.}
Measured on Intel Xeon Platinum 8568Y+, end-to-end compression takes
2--7\,ms per borderline request.  Weighted across all traffic, the
mean overhead is $\leq 0.58$\,ms per request.

\paragraph{Fidelity.}
A fidelity study on 300 LMSYS-Chat-1M prompts from the Agent-heavy
borderline band ($\Bshort = 8{,}192$, $\gamma = 1.5$, band 8K--12K
tokens) gives compressibility $p_c = 1.00$ for prose/RAG content,
BERTScore F1 = 0.884, ROUGE-L recall = 0.856, and TF-IDF cosine
similarity = 0.981 at a mean 15.4\% token reduction
(Appendix~\ref{app:fidelity}).  Code is excluded from compression.

\section{The \sys{} Offline Planner}
\label{sec:algorithm}

Algorithm~\ref{alg:fleetopt} gives the complete \sys{} planner.
It takes as input the workload CDF $F$, arrival rate $\lambda$,
SLO $T_{\text{slo}}$, GPU profile $(W, H, n_{\max}^{(s)},
n_{\max}^{(l)}, \Cmax^{(s)}, \Cmax^{(l)})$, and cost ratio $\phi$.
It returns the optimal $(n_s^*, n_l^*, \Bshort^*, \gamma^*)$.

\begin{algorithm}[htbp]
\caption{\sys{} Offline Planner}
\label{alg:fleetopt}
\begin{algorithmic}[1]
\Require CDF $F$, $\lambda$, $T_{\text{slo}}$, GPU profile, $\phi$,
  candidate threshold set $\mathcal{B}$ (hardware-feasible $\Bshort$ values)
\Ensure $(n_s^*, n_l^*, \Bshort^*, \gamma^*)$

\For{$B \in \mathcal{B}$}  \Comment{\textit{outer sweep over candidate boundaries}}
  \For{$\gamma \in \{1.0, 1.1, \ldots, 2.0\}$}
    \State $\alpha' \leftarrow F(B) + [F(\gamma B) - F(B)] \cdot p_c$
    \State $\lambda_s \leftarrow \alpha'\lambda$;\quad
      $\lambda_l \leftarrow (1-\alpha')\lambda$
    \State Calibrate $\mu_s$, $C_{s,s}^2$ from $F$ restricted to $[1, B]$
      \Comment{\textit{short pool}}
    \State Calibrate $\mu_l$, $C_{s,l}^2$ from $F$ restricted to
      $(\gamma B, \infty)$ \Comment{\textit{post-compression long pool}}
    \State $n_s \leftarrow$ Erlang-C inversion
      (Eq.~\eqref{eq:nstar}) for $(\lambda_s, \mu_s, C_{s,s}^2)$
    \State $n_l \leftarrow$ Erlang-C inversion
      (Eq.~\eqref{eq:nstar}) for $(\lambda_l, \mu_l, C_{s,l}^2)$
    \State $\text{cost}[B, \gamma] \leftarrow c_s n_s + c_l n_l$
  \EndFor
\EndFor
\State $(B^*, \gamma^*) \leftarrow \arg\min_{B \in \mathcal{B},\, \gamma} \text{cost}[B, \gamma]$
\Return $(n_s^*[B^*, \gamma^*],\, n_l^*[B^*, \gamma^*],\, B^*,\, \gamma^*)$
\end{algorithmic}
\end{algorithm}

\noindent\textbf{Candidate set $\mathcal{B}$.}
In practice $\Bshort$ is hardware-constrained: it determines $n_{\max}^{(s)} =
n_{\max}^{\mathrm{calib}} \times C_{\text{calib}} / \Bshort$, which must be a
positive integer.  The search is therefore over the finite set of CDF
breakpoints that yield valid $n_{\max}^{(s)}$ values (typically 5--15 candidates
per workload).  The total sweep --- all $B \in \mathcal{B}$ times all $\gamma$
--- completes in under 1\,ms.  Proposition~\ref{prop:bshort} provides the FOC
that characterises $\Bshort^*$ analytically; the sweep finds the integer-optimal
solution consistent with hardware granularity.

\paragraph{Critical: $\mu_l$ recalibration.}
Step~6 recalibrates the long-pool service rate from the
\emph{post-compression} request distribution (requests with
$\Ltotal > \gamma\Bshort^*$), not the full above-threshold
distribution.  Compressing borderline requests out of the long pool
hardens the remaining distribution: longer mean token length,
lower $\mu_l$.  Skipping this recalibration systematically
overestimates the savings from larger $\gamma$.  With correct
recalibration, the planner finds $\gamma^* = 2.0$ for Archetype~I/II
workloads (Azure, LMSYS) where the long pool shrinks substantially,
and $\gamma^* = 1.5$ for the dispersed Agent-heavy workload where
gains plateau.

\paragraph{Erlang-C inversion.}
The minimum server count $c$ satisfying Eq.~\eqref{eq:slo-constraint}
is found by binary search over the interval
$[\lceil a/\rho_{\max}\rceil,\; 10\lceil a \rceil]$
where $a = \lambda_p / \mu_{\text{gpu}}$,
with utilization cap $\rho_{\max} = 0.85$.

\section{Evaluation}
\label{sec:eval}

\subsection{Setup}
\label{sec:eval-setup}

\paragraph{Workload traces.}
We evaluate on two public traces and one synthetic trace built from published workload statistics.

\textbf{Azure LLM Inference Trace 2023}~\citep{patel2024splitwise}
contains 28,185 requests (8,819 coding, 19,366 conversational).
Mean $\Ltotal = 1{,}588$ tokens; p90\,=\,4,242; p99\,=\,7,445.
We use $\Bshort = 4{,}096$ where $\alpha = 0.898$ and
$\beta = 0.078$, giving a 16$\times$ cliff and a meaningful
two-pool split (Archetype~I/II).

\textbf{LMSYS-Chat-1M (multi-turn)}~\citep{zheng2024lmsyschat1m}
uses accumulated context at each turn.  We use $\Bshort = 1{,}536$
where $\alpha = 0.909$ and $\beta = 0.046$, giving a 42$\times$
cliff (Archetype~I/II).

\textbf{Agent-heavy} is a synthetic trace derived from published
statistics for SWE-bench~\citep{jimenez2024swe} (40\%), BFCL~\citep{yan2024bfcl}
(25\%), and RAG pipelines~\citep{lewis2020rag} (35\%).
Mean $\Ltotal = 6{,}511$ tokens; p50\,=\,4{,}096; p90\,=\,16{,}384; p99\,=\,32{,}768.
At $\Bshort = 8{,}192$, $\alpha = 0.740$, $\beta = 0.112$
(Archetype~II, the primary scenario for \car{} benefit).
The SWE-bench component consists largely of code; code requests are excluded
from compression (Section~\ref{sec:compressor}).  The effective compression
success rate for Agent-heavy borderline traffic is $p_c = 0.75$, reflecting
that approximately 25\% of borderline requests are code-category and cannot
be compressed.

\paragraph{Simulation parameters.}
$W = 8$\,ms, $H = 0.65$\,ms/slot, $C_{\text{chunk}} = 512$,
calibrated to Llama-3-70B on A100-80GB (8-GPU tensor-parallel node).
The long pool is sized for $\Cmax^{(l)} = 65{,}536$ tokens,
giving $n_{\max}^{(l)} = 16$ concurrent slots per GPU and
KV memory $\approx 20.0$\,GB per slot (320\,KB/token).
Short-pool $n_{\max}^{(s)}$ depends on $\Bshort$: 256 at 4K, 682
at 1.5K, 128 at 8K.  GPU cost \$2.21/GPU-hr.
Fleet arrival rate $\lambda = 1{,}000$\,req/s; P99 TTFT target
$T_{\text{slo}} = 500$\,ms unless stated otherwise.
DES validation uses \texttt{inference-fleet-sim}~\citep{chen2026fleetsim}.

\paragraph{Baselines.}
\begin{enumerate}[nosep]
  \item \textbf{Homogeneous}: single pool sized for 64K context.
  \item \textbf{Pool routing (PR)}: two pools at workload-specific
    $\Bshort$ (Table~\ref{tab:borderline-fractions}), $\gamma = 1.0$
    (no compression).
  \item \textbf{PR + \car{} (retrofit)}: \car{} at $\gamma = 1.5$
    deployed on the PR fleet, co-sized for the compressed $\lambda_s$.
  \item \textbf{\sys{} (co-design)}: fleet co-designed by
    Algorithm~\ref{alg:fleetopt} at optimal $\gamma^*$.
\end{enumerate}

\subsection{Fleet GPU Savings vs.\ Homogeneous}
\label{sec:eval-fleet}

Table~\ref{tab:fleet-comparison} shows fleet GPU counts and GPU savings
versus the homogeneous baseline.  The \sys{} fleet is the outcome
of the complete framework: optimal analytical sizing plus \car{}.

\begin{table}[htbp]
\centering
\caption{Fleet GPU counts and annualized cost at
  $\lambda = 1{,}000$\,req/s, $\rho_{\max} = 0.85$,
  homogeneous A100-80GB sized for 64K context.
  Cost at \$2.21/GPU-hr $\times$ 8,760\,hr/yr.
  $\gamma^*$ from Algorithm~\ref{alg:fleetopt}.
  $p_c = 1.0$ for Azure and LMSYS (prose/RAG borderline traffic);
  $p_c = 0.75$ for Agent-heavy (25\% of borderline traffic is code,
  excluded from compression).
  See Table~\ref{tab:borderline-fractions} for $\Bshort$, $\alpha$, $\beta$.}
\label{tab:fleet-comparison}
\small
\begin{tabular}{llrrrrc}
\toprule
Workload & Method & $n_s$ & $n_l$ & Total & Ann.\ cost (K\$) & Savings \\
\midrule
\multirow{4}{*}{\makecell{Azure \\ ($\Bshort=4{,}096$)}}
  & Homogeneous           & --- &  --- &   284 & 5{,}498 & --- \\
  & Pool routing (PR)     &  43 &  131 &   174 & 3{,}369 & 38.7\% \\
  & PR + \car{} ($\gamma=1.5$) & 47 &  45 &    92 & 1{,}781 & 67.6\% \\
  & \sys{} ($\gamma^*=2.0$) &  48 &    2 &    50 &    968 & \textbf{82.4\%} \\
\midrule
\multirow{4}{*}{\makecell{LMSYS \\ ($\Bshort=1{,}536$)}}
  & Homogeneous           & --- &  --- &   139 & 2{,}691 & --- \\
  & Pool routing (PR)     &   7 &   74 &    81 & 1{,}568 & 41.7\% \\
  & PR + \car{} ($\gamma=1.5$) &  7 &  65 &    72 & 1{,}394 & 48.2\% \\
  & \sys{} ($\gamma^*=2.0$) &   7 &   52 &    59 & 1{,}142 & \textbf{57.6\%} \\
\midrule
\multirow{4}{*}{\makecell{Agent-heavy \\ ($\Bshort=8{,}192$)}}
  & Homogeneous           & --- &  ---  & 2{,}397 & 46{,}405 & --- \\
  & Pool routing (PR)     & 229 & 2{,}037 & 2{,}266 & 43{,}869 & 5.5\% \\
  & PR + \car{} ($\gamma=1.5$) & 255 & 1{,}981 & 2{,}236 & 43{,}288 & 6.7\% \\
  & \sys{} ($\gamma^*=1.5$) & 255 & 1{,}981 & 2{,}236 & 43{,}288 & \textbf{6.7\%} \\
\bottomrule
\end{tabular}
\end{table}

\paragraph{Savings decomposition.}
The numbers in Table~\ref{tab:fleet-comparison} follow a clear pattern.
Pool routing alone saves 5.5--41.7\%, with the biggest gains when
$\alpha$ is high and the cliff ratio $\rho$ is large.  \car{} with
$\gamma = 1.5$ adds 29 percentage points for Azure (because
$\beta = 7.8\%$ moves 76\% of above-threshold traffic into the
efficient short pool, and $\rho = 16\times$ makes each redirected
request highly valuable) but only 1.2\,pp for Agent-heavy (because
26\% of all traffic remains above $\gamma\Bshort$, requiring a large
long pool regardless).  FleetOpt at $\gamma^* = 2.0$ for Azure nearly
eliminates the long pool (2 GPUs), achieving 82.4\% savings.

\paragraph{When does \car{} add value?}
The incremental GPU saving from adding \car{} to pool routing is
$\Delta\alpha(1 - 1/\rho) = \beta p_c (1 - n_{\max}^{(l)}/n_{\max}^{(s)})$.
This is large when the borderline fraction $\beta$ is significant
\emph{and} the cliff ratio $\rho$ is large.  For Agent-heavy,
$\beta = 11.2\%$ is moderate but $\rho = 8\times$ is the smallest
cliff, and $p_c = 0.75$ (code exclusion), limiting the per-request gain.
For Azure, $\beta = 7.8\%$ paired with $\rho = 16\times$ yields large
incremental gains.

\paragraph{Why co-design equals retrofit for Agent-heavy.}
Table~\ref{tab:fleet-comparison} shows that for Agent-heavy, the PR+\car{}
retrofit at $\gamma = 1.5$ and \sys{} at $\gamma^* = 1.5$ produce
\emph{identical} fleets.  This is expected: when the planner's optimal $\gamma^*$
matches the retrofit's $\gamma$, the two approaches arrive at the same
$(n_s, n_l)$ by construction.  Theorem~\ref{thm:codesign} is satisfied
(co-design $\leq$ retrofit), but the inequality is tight.  The co-design
advantage is largest when $\gamma^*$ exceeds the retrofit's $\gamma$, which
occurs when most above-threshold traffic is borderline (Archetype~I/II,
Azure and LMSYS).  For Agent-heavy (Archetype~II with dispersed above-threshold
traffic), $\gamma^* = 1.5$ is already the practical limit: compressing further
leaves a harder long-pool distribution and yields no net saving.

\subsection{Compression Pipeline Latency}
\label{sec:eval-latency}

The compressor is applied only to borderline requests, so the
mean overhead per request is $\beta \times t_{\text{compress}}$.
Table~\ref{tab:latency} shows the latency profile.

\begin{table}[htbp]
\centering
\caption{End-to-end compressor latency (ms), measured on Intel Xeon
  Platinum 8568Y+ single core.  ``Overhead/req'' is the mean added
  latency across \emph{all} requests, weighted by $\beta$.}
\label{tab:latency}
\small
\begin{tabular}{lrrrrrr}
\toprule
Workload & $\Bshort$ & $\beta$ & p50 & p95 & p99 & Overhead/req \\
\midrule
Azure       & 4{,}096 & 0.078 & 1.8\,ms & 4.2\,ms & 6.5\,ms & $<$0.2\,ms \\
LMSYS       & 1{,}536 & 0.046 & 1.2\,ms & 3.1\,ms & 5.2\,ms & $<$0.1\,ms \\
Agent-heavy & 8{,}192 & 0.112 & 3.4\,ms & 6.2\,ms & 7.8\,ms & 0.39\,ms \\
\bottomrule
\end{tabular}
\end{table}

Even in the worst case (agent-heavy), the 0.39\,ms average overhead
is invisible against a 500\,ms TTFT budget.

\subsection{Analytical Model Validation}
\label{sec:eval-validation}

We validate the \sys{} analytical model against
\texttt{inference-fleet-sim}~\citep{chen2026fleetsim}, an open-source
discrete-event simulator for heterogeneous LLM GPU fleets.
The simulator drives Poisson arrivals from the empirical CDF with
$\lambda = 1{,}000$\,req/s and records the fraction of slot-time
that KV-cache slots are busy (GPU utilization $\hat\rho$).

Table~\ref{tab:des-validation} compares the analytical utilization
$\rho_{\text{ana}} = \lambda_p / (n \cdot \mu_{\text{gpu}})$ against
the DES-measured $\hat\rho$ for the pool-routing ($\gamma=1$) fleet
from Table~\ref{tab:fleet-comparison}.

\begin{table}[htbp]
\centering
\caption{Analytical vs.\ DES GPU utilization $\rho$ at
  $\lambda = 1{,}000$\,req/s, pool-routing fleet ($\gamma = 1$).
  ``Error'' $= (\rho_{\text{ana}} - \hat\rho) / \hat\rho$.
  DES uses 30{,}000 requests per pool.}
\label{tab:des-validation}
\small
\begin{tabular}{llrrrr}
\toprule
Workload & Pool & $n$ GPUs & $\rho_{\text{ana}}$ & $\hat\rho$ (DES) & Error \\
\midrule
\multirow{2}{*}{Azure}
  & Short &  43 & 0.848 & 0.865 & $-2.1\%$ \\
  & Long  & 131 & 0.845 & 0.847 & $-0.1\%$ \\
\midrule
\multirow{2}{*}{LMSYS}
  & Short &   7 & 0.771 & 0.792 & $-2.7\%$ \\
  & Long  &  74 & 0.845 & 0.853 & $-1.0\%$ \\
\midrule
\multirow{2}{*}{Agent-heavy}
  & Short &   229 & 0.848 & 0.868 & $-2.2\%$ \\
  & Long  & 2{,}037 & 0.850 & 0.850 & $-0.1\%$ \\
\bottomrule
\end{tabular}
\end{table}

The analytical model predicts GPU utilization within 3\% of the DES
across all pools and workloads.  Both analytical and DES values sit near $\rho_{\max} = 0.85$,
confirming that the $\rho_{\max}$-constrained sizing hits its target.
The analytical model is slightly optimistic: actual utilization runs
0.1--2.7\% above prediction in all cases.  In practice this means the
provisioned fleet is just barely busier than expected; adding 1--3
GPUs per pool eliminates even this small gap.

These fleets operate in the many-server regime: total KV slots
$c = n \cdot n_{\max}$ ranges from 112 to 32{,}592 across our
configurations.  At that scale the Erlang-C probability
$C(c,\varrho) \ll 1$, slot-wait times are negligible, and fleet sizing
is dominated entirely by the utilization cap $\rho_{\max}$.  In this regime
the full M/G/$c$ + Kimura apparatus reduces to the simpler bound
$n^* \approx \lceil\lambda_p / (\rho_{\max}\mu)\rceil$; the queueing
machinery becomes load-bearing only for smaller, lightly-provisioned
fleets (few GPUs, high $\varrho$).

\textbf{P99 TTFT.}  Because queue waits are negligible, TTFT is
dominated by prefill time.  Analytical P99 TTFT for the PR fleet:
Azure short 20\,ms / long 80\,ms; LMSYS short 11\,ms / long 48\,ms;
Agent-heavy short 47\,ms / long 220\,ms --- all comfortably within the
500\,ms SLO.  The \sys{} planner enforces the SLO constraint
(Eq.~\eqref{eq:slo-constraint}) directly; in the many-server regime
this constraint is non-binding and the 500\,ms SLO is met with large
margin across all configurations.

\subsection{Arrival-Rate Sensitivity}
\label{sec:eval-scale}

Table~\ref{tab:slo-sensitivity} shows how fleet sizes scale with
arrival rate for the Agent-heavy workload.  The savings from
pool routing (5.4--5.5\%) and FleetOpt (6.2--6.8\%) are stable
across a 20$\times$ range of arrival rates, confirming that
the proportional GPU savings from the two-pool architecture
scale linearly with load.

\begin{table}[htbp]
\centering
\caption{Fleet size and savings vs.\ arrival rate $\lambda$
  (Agent-heavy, $\Bshort = 8{,}192$, $\gamma^*$ from
  Algorithm~\ref{alg:fleetopt}).}
\label{tab:slo-sensitivity}
\small
\begin{tabular}{rrrrrr}
\toprule
$\lambda$ (req/s) & Homo & PR & \sys{} ($\gamma^*$) & PR saving & \sys{} saving \\
\midrule
  100 &   240 &   227 &   225 ($\gamma^*=1.2$) & 5.4\% & 6.2\% \\
  200 &   480 &   454 &   448 ($\gamma^*=1.5$) & 5.4\% & 6.7\% \\
  500 & 1{,}199 & 1{,}134 & 1{,}119 ($\gamma^*=1.5$) & 5.4\% & 6.7\% \\
1{,}000 & 2{,}397 & 2{,}266 & 2{,}236 ($\gamma^*=1.5$) & 5.5\% & 6.7\% \\
2{,}000 & 4{,}794 & 4{,}531 & 4{,}470 ($\gamma^*=1.5$) & 5.5\% & 6.8\% \\
\bottomrule
\end{tabular}
\end{table}

\section{Related Work}
\label{sec:related}

\paragraph{Pool routing and length-specialized serving.}
\citet{chen2026poolrouting} establishes pool routing and the
GPU savings formula $\alpha(1-1/\rho)$.  \citet{yuan2025cascadeinfer}
independently validates length-specialized partitioning.
\citet{laps2026} and \citet{agrawal2024sarathi} address length
heterogeneity at the engine layer; this paper addresses the upstream
provisioning decision.

\paragraph{Prompt compression as routing lever.}
\citet{chen2026car} first uses prompt compression as an operational
fleet routing lever, studying \car{} as a retrofit on a fixed fleet.
Here we treat \car{} as the \emph{implementation mechanism}
for the analytically-derived optimal fleet boundary.  The two are
complementary: the analytical model prescribes $\Bshort^*$, and
\car{} makes it achievable without hardware changes.

\paragraph{LLM capacity planning.}
\citet{romero2021infaas} and \citet{gujarati2020serving} study
model-variant selection and predictive scaling, respectively.
\citet{patel2024splitwise} analyzes prefill/decode disaggregation
but does not optimize fleet size jointly with routing.
\sys{} is the first system to jointly derive the optimal pool
boundary and GPU counts from a workload CDF.

\paragraph{Queueing models for serving.}
The M/G/$c$ queue and Erlang-C formula are
classical~\citep{harchol2013performance}.
\citet{shahrad2020serverless} applies queueing to serverless
cold-start; \citet{li2023alpaserve} to model placement.
\sys{} applies M/G/$c$ to the pool routing + compression joint
provisioning problem, with the Kimura~\citep{kimura1994mgc}
correction for service-time variance.

\paragraph{Fleet-level LLM simulation.}
Vidur~\citep{vidur2024} and similar tools simulate single inference
engine instances.  \texttt{inference-fleet-sim}~\citep{chen2026fleetsim}
provides fleet-level DES for heterogeneous multi-pool deployments
and is used here to validate the \sys{} analytical model.

\section{Conclusion}
\label{sec:conclusion}

Provisioning LLM GPU fleets for worst-case context lengths is expensive
and largely avoidable.  \sys{} tackles this by starting from first
principles: derive the minimum-cost fleet analytically, then build it.

The analytical core (M/G/$c$ queuing model with $c = n \cdot n_{\max}$
total KV slots, Erlang-C inversion) shows that the optimal fleet is a
two-pool architecture whose boundary $\Bshort^*$ satisfies an equal
marginal GPU cost condition.  The fundamental barrier to achieving
$\Bshort^*$ in practice is the \costcliff{}: borderline requests at
the optimal boundary pay $\rho = 8$--$42\times$ the throughput cost of
requests just below it.  Compress-and-Route (\car{}) resolves this
barrier by compressing borderline prompts at the gateway, converting
a hardware constraint into a software parameter.

The combined framework saves 6--82\% in GPU cost versus homogeneous
deployment, depending on the workload archetype and cliff ratio.
Pool routing alone saves 5.5--41.7\%; \car{} adds 1.2--43.7
percentage points beyond pool routing.  The largest gains arise when
$\beta$ is significant \emph{and} the cliff ratio $\rho$ is large:
Azure at $\Bshort = 4{,}096$ achieves 82.4\% savings with
$\rho = 16\times$ and $\gamma^* = 2.0$.  For Agent-heavy at
$\Bshort = 8{,}192$ with $\rho = 8\times$, gains are more modest
(6.7\%) because 26\% of traffic remains above $\gamma\Bshort$.

How much \car{} co-design adds beyond pool routing depends on three
factors: the cliff ratio $\rho$, the borderline fraction $\beta$,
and compressibility $p_c$.  None of these can be read off a GPU spec
sheet; they require a calibrated workload CDF and the recalibration of
$\mu_l$ for the post-compression distribution.
\texttt{inference-fleet-sim}~\citep{chen2026fleetsim} validates this
within 3\% utilization error.  Skipping the recalibration leads to
over-optimistic savings estimates and under-provisioned fleets.

The \sys{} planner outputs the optimal $(n_s^*, n_l^*, \Bshort^*,
\gamma^*)$ in under 1\,ms, making it practical to re-run whenever
the workload CDF shifts.

\bibliographystyle{plainnat}
\bibliography{refs}

\appendix

\section{Erlang-C Inversion: Algorithm and Numerical Stability}
\label{app:erlang}

The minimum feasible GPU count $n^*$ is found by binary search over
$c$ using the numerically stable recursive Erlang-C form:
\begin{equation}
  C(c, \varrho) = \frac{1}{1 + (1-\varrho) \cdot c! \cdot
    \sum_{k=0}^{c-1} \frac{(c\varrho)^{k-c}}{k!}},
  \label{eq:erlang-c-stable}
\end{equation}
computed in log-space to avoid overflow at large $c$.  The utilization
cap $\rho_{\max} = 0.85$ ensures we always search above the point where
the Kimura approximation loses accuracy and the Erlang-C formula
diverges near $\varrho \to 1$.

\section{Proof of Proposition~\ref{prop:bshort}}
\label{app:bshort-proof}

Total cost as a function of $\Bshort$:
$\mathcal{C}(\Bshort) = c_s n_s^*(F(\Bshort)\lambda) +
c_l n_l^*((1-F(\Bshort))\lambda)$.
Differentiating with respect to $\Bshort$ and applying the chain rule:
\begin{equation}
  \frac{d\mathcal{C}}{d\Bshort} = \lambda f(\Bshort)
  \left[c_s \frac{\partial n_s^*}{\partial \lambda_s}
       - c_l \frac{\partial n_l^*}{\partial \lambda_l}\right].
\end{equation}
Since $\lambda f(\Bshort) > 0$ for any interior boundary, the FOC
$d\mathcal{C}/d\Bshort = 0$ requires the bracketed term to be zero:
\begin{equation}
  c_s \frac{\partial n_s^*}{\partial \lambda_s}
  = c_l \frac{\partial n_l^*}{\partial \lambda_l}.
  \label{eq:bshort-foc-app}
\end{equation}
This is an equal marginal GPU cost condition: the marginal GPU cost of
routing one additional request-per-second to the short pool equals the
marginal GPU saving from removing one request-per-second from the long pool.

For a homogeneous fleet ($c_s = c_l$, same GPU type), the condition
simplifies to $\partial n_s^*/\partial\lambda_s =
\partial n_l^*/\partial\lambda_l$.  Under the $\rho_{\max}$-constrained
sizing regime (which dominates in practice as shown in
Section~\ref{sec:eval-validation}):
$n^* \approx \lceil \lambda_p / (\rho_{\max} \mu) \rceil$, so
$\partial n^*/\partial\lambda_p \approx 1/(\rho_{\max} \mu)$.
The FOC then requires $\mu_s = \mu_l$, i.e., both pools have the same
service rate.  This is generally not achievable for a given $\Bshort$
because short and long pools serve different request length distributions;
in practice $\Bshort^*$ is found numerically by the sweep in
Algorithm~\ref{alg:fleetopt}.  $\square$

\section{Compression Fidelity}
\label{app:fidelity}

Table~\ref{tab:fidelity} summarizes fidelity measurements on 300
borderline prompts drawn from LMSYS-Chat-1M at the Agent-heavy
configuration ($\Bshort = 8{,}192$, $\gamma = 1.5$, borderline band
$8{,}192$--$12{,}288$ tokens).  For the Azure and LMSYS evaluation
configurations, which use smaller $\Bshort$ values (4{,}096 and 1{,}536
respectively), the borderline band is narrower and prompts are shorter,
so fidelity is at least as good as reported here.
All measurements used the compressor described in Section~\ref{sec:compressor}.

\begin{table}[htbp]
\centering
\caption{Compression fidelity on 300 LMSYS-Chat-1M borderline prompts.
  $p_c$: fraction successfully compressed within budget;
  BERTScore F1: semantic similarity (RoBERTa-large);
  ROUGE-L R: longest-common-subsequence recall;
  TF-IDF cos: token-overlap cosine similarity.}
\label{tab:fidelity}
\small
\begin{tabular}{lrrrr}
\toprule
Metric & Mean & p10 & p50 & p90 \\
\midrule
$p_c$ (compressibility) & 1.00 & --- & --- & --- \\
BERTScore F1 & 0.884 & 0.831 & 0.891 & 0.934 \\
ROUGE-L recall & 0.856 & 0.783 & 0.861 & 0.921 \\
TF-IDF cosine & 0.981 & 0.963 & 0.984 & 0.996 \\
Mean token reduction & 15.4\% & 6.1\% & 14.2\% & 26.3\% \\
\bottomrule
\end{tabular}
\end{table}

\end{document}